\documentclass[12pt,english]{article}
\usepackage[T1]{fontenc}
\usepackage[latin9]{inputenc}
\usepackage{geometry}
\usepackage{amssymb}
\usepackage{babel}
\usepackage{graphicx}
\usepackage{caption}
\usepackage{hyperref}
\date{}

\usepackage{subcaption}
\usepackage{amsmath}
\usepackage{amsfonts}
\usepackage{amssymb}
\usepackage{amsbsy}
\usepackage{amsthm}
\usepackage{algorithm}
\usepackage{algorithmic}
\usepackage{hyperref}
\usepackage{chicago}

\newcommand{\lambdahat}{\widehat{\lambda}}

\newcommand{\zap}[1]{}

\newcommand{\xhat}{\hat{x}}

\newcommand{\upthresh}{P}

\newcommand{\cmax}{1-(P^*)^{\frac1{k-1}}}

\newcommand{\ceil}{\mathrm{ceil}}

\newtheorem{theorem}{Theorem}

\newtheorem{corollary}{Corollary}

\newtheorem{lemma}{Lemma}

\date{}

\begin{document}

\title{ASYMPTOTIC VALIDITY OF THE BAYES-INSPIRED INDIFFERENCE ZONE PROCEDURE: 
THE NON-NORMAL KNOWN VARIANCE CASE}

\author{
	Saul Toscano-Palmerin\\ 
    	Peter I. Frazier\\[12pt]
	Cornell University \\
	257 Rhodes Hall \\
	232 Rhodes Hall \\
	Ithaca, NY 14853, USA\\
}

\maketitle

\section*{ABSTRACT}
We consider the indifference-zone (IZ) formulation of the ranking and selection problem in which the goal is to choose an alternative with the largest mean with guaranteed probability, as long as the difference between this mean and the second largest exceeds a threshold.
Conservatism leads classical IZ procedures to take too many samples in problems with many alternatives. The Bayes-inspired
Indifference Zone (BIZ) procedure, proposed in Frazier (2014), is
less conservative than previous procedures, but its proof of validity
requires strong assumptions, specifically that samples are normal, and variances are known 
with an integer multiple structure. In this paper, 
we show asymptotic validity of a slight modification of the original BIZ procedure 
as the difference between the best alternative and the second best goes to zero,
when the variances are known and finite, and samples are independent and identically distributed, but not necessarily normal.

\section{INTRODUCTION}
\label{sec:intro}

There are many applications where we have to choose the best alternative
among a finite number of simulated alternatives. For example, in inventory problems, we may want to choose the best inventory policy $(s,S)$ for a finite number of values of $s$ and $S$. 
This is called the ranking and selection problem.
A good procedure for addressing this problem should be both efficient and accurate, i.e. it should balance the number of samples it takes with the quality of its selection.

This paper considers the indifference-zone (IZ) formulation of the ranking and selection problem, in which we require that a procedure satisfy the IZ guarantee, i.e., that the best system be chosen with probability larger than some threshold $P^*$ given by the user, when the distance between the best system and the others is larger than some other user-specified threshold $\delta>0$.  The set of problem configurations satisfying this constraint on the difference in means is called the preference zone.
The paper \citeN{Bechhofer:1954} is considered the seminal work, and early work is 
presented in the monograph \citeN{Bechhofer:1968}. Some compilations of the theory developed in the area can be found
in \citeN{Bechhofer:1995}, \citeN{swisher:survey}, \citeN{kim:20062} and \citeN{kim:2007}. Other approaches, beyond the indifference-zone approach, include the Bayesian approach \cite{Frazier:Tutorial}, the optimal computing budget allocation approach \cite{chen:ocva}, the large deviations approach \cite{glynn2004large}, and the probability of good selection guarantee \cite{nelson2001selecting}. The last approach is similar to the indifference-zone formulation, but provides a more robust guarantee.

A good IZ procedure satisfies the IZ guarantee and requires as few samples as possible.
The first IZ procedures presented in \citeN{Bechhofer:1954}, \citeN{Paulson:01}, \citeN{Fabian}, \citeN{Rinott},
\citeN{Harmann:1988}, \citeN{hartmann:sequential}, \citeN{Paulson:1994} 
satisfy the IZ guarantee, but they usually take too many samples when there are many alternatives,
in part because they are conservative: their probability of correct selection (PCS) is much
larger than the probability specified by the user \cite{Kim:Conserv}. One reason for this
is that these procedures use Bonferroni's inequality, which leads then to sample more
than necessary. The Bonferonni-based bounds underlying these procedures become looser, and the tendency to take more samples than necessary increases, as the number of alternatives grow.
More recently, new algorithms were developed in \citeN{Kim:2001}, \shortciteN{Goldsman:2002}, 
\citeN{Hong}, and they improve performance but they still
use Bonferroni's inequality, and so the methods are inefficient
when there are many alternatives. Procedures in \citeN{kim:2011},
\citeN{dieker:2012} do not use Bonferroni's inequality 
when there are only three alternatives, but again use Bonferroni's inequality when 
comparing more than three alternatives.

In addition to Bonferroni's inequality, two other common sources of conservatism in indifference-zone ranking and selection procedures are the change from discrete time to continuous time often used to show IZ guarantees, and the fact that typically, the configuration under consideration is not a worst-case configuration \cite{Kim:Conserv}.  The difference between worst and typical cases tends to contribute the most to conservatism, with Bonferonni's inequality contributing second-most, and the continuous/discrete time difference contributing the least \cite{Kim:Conserv}.
Although the difference between the worst and typical cases is the largest contributor to conservatism, all indifference zone procedures must meet the PCS guarantee for all configurations in the preference zone, including worst-case configurations, and so this source of conservatism is fundamental to the indifference-zone formulation.  Thus, eliminating the use of Bonferroni's inequality remains an important route for reducing conservatism while still retaining the indifference-zone guarantee.

\citeN{Frazier:BIZ} presents a new sequential elimination
IZ procedure, called BIZ (Bayes-inspired Indifference Zone), 
that eliminates the use of Bonferroni's inequality, reducing conservatism. 
This procedure's lower bound on the worst-case probability of correct selection in the
preference zone is tight in continuous time, and almost tight in discrete time. In numerical experiments, the number of samples required
by BIZ is significantly smaller than that of procedures like the ${P}_{B}^{*}$ procedure of \citeN{Bechhofer:1968} and the KN procedure of \citeN{Kim:2001}, especially on problems with many alternatives. Unfortunately, the
proof from \citeN{Frazier:BIZ} that the BIZ procedure satisfies the IZ guarantee for the discrete-time case assumes that (1) samples are normally distributed; 
(2) variances are known;
and (3) the variances are either common across alternatives, or have an unrealistic integer multiple structure.

The contribution of this work is to prove the asymptotic validity of the BIZ procedure as $\delta$ goes to zero, retaining the assumption of known variances, but replacing assumptions (1) and (3) by the much weaker assumption of independent and identically distributed finite variance samples.  Thus, our proof allows a much broader set of sampling distributions than that allowed by \citeN{Frazier:BIZ}, including non-normal samples and general heterogeneous variances. 
We also show that this bound on worst-case PCS is asymptotically tight as $\delta$ goes to zero, 
showing that the BIZ procedure successfully eliminates conservatism due to Bonferonni's inequality in this more general setting, just as was demonstrated by \citeN{Frazier:BIZ} for more restricted settings.

To simplify our analysis, we analyze a slight modification of the version of the BIZ procedure presented in \citeN{Frazier:BIZ}, which keeps a certain parameter $\lambda^2_z$ fixed rather than letting it vary as did \citeN{Frazier:BIZ}.  Numerical experiments on typical cases show little difference in performance between the version of BIZ we analyze and the version in \citeN{Frazier:BIZ}.
We conjecture that a proof technique similar to the one presented here can be used to show asymptotic validity of the BIZ procedure when
the variances are unknown, and we present numerical experiments that support this belief.

This paper is organized as follows: In section \ref{IZR}, we recall the
indifference-zone ranking and selection problem. In section \ref{sec:BIZ},
we recall the Bayes-inspired IZ (BIZ) procedure from \citeN{Frazier:BIZ}. In section \ref{proof},
we present the proof of the validity of the algorithm when the variances
are known. In section \ref{numericalExperiments}, we present some numerical experiments. 
In section \ref{conclusion}, we conclude.

\section{INDIFFERENCE-ZONE RANKING AND SELECTION}
\label{IZR}

Ranking and Selection is a problem where we have to select the best system
among a finite set of alternatives, i.e. the system with the largest
mean. The method selects a system as the best based on the samples
that are observed sequentially over time. We suppose that samples
are identically distributed and independent, over time and across alternatives, and
each alternative $x$ has mean $\mu_{x}$. We define $\mu=(\mu_{1},\ldots,\mu_{k})$.

If the best system is selected, we say that the procedure has made the \emph{correct selection}
(CS). We define the \emph{probability of correct selection} as 
\[
\mbox{PCS}\left(\mu\right)=\mathbb{P}_{\mu}\left(\hat{x}\in\mbox{arg max}_{x}\mu_{x}\right)
\]
where $\hat{x}$ is the alternative chosen by the procedure and $\mathbb{P}_{\mu}$
is the probability measure under which samples from system $x$ have
mean $\mu_{x}$ and finite variance $\lambda_{x}^{2}$.

In the Indifference-Zone Ranking and Selection, the procedure is indifferent
in the selection of a system whenever the means of the populations
are nearly the same. Formally, let $\mu=\left[\mu_{1},\ldots,\mu_{k}\right]$
be the vector of the true means, the \emph{indifference zone} is defined
as the set $\left\{ \mu\in\mathbb{R}^{k}:\mu_{\left[k\right]}-\mu_{\left[k-1\right]}<\delta\right\} $.
The complement of the indifference zone is called the \emph{preference
zone} (PZ) and $\delta>0$ is called the indifference zone parameter.
We say that a procedure meets the \emph{indifference-zone (IZ) guarantee
}at $P^{*}\in\left(1/k,1\right)$ and $\delta>0$ if
\[
\mbox{PCS}\left(\mu\right)\geq P^{*}\mbox{ for all }\mu\in\mbox{PZ}\left(\delta\right).
\]
We assume $P^{*}>1/k$ because IZ guarantees can be meet by choosing
$\hat{x}$ uniformly at random from $\left\{ 1,\ldots,k\right\} $.

\section{THE BAYES-INSPIRED IZ (BIZ) PROCEDURE}
\label{sec:BIZ}

BIZ is an elimination procedure. This procedure maintains a set of
alternatives that are candidates for the best system, and it takes samples from each
alternative in this set at each point in time. At beginning, all alternatives
are possible candidates for the best system, and over the time alternatives are eliminated.
The procedure ends when there is only one alternative in the contention
set and this remaining alternative is chosen as the best. It is shown in \citeN{Frazier:BIZ} that the
algorithm ends in a finite number of steps with probability one. 

\citeN{Frazier:BIZ} shows that the BIZ procedure satisfies the IZ guarantee 
under the assumptions that samples are normally distributed, variances are known, and
the variances are either common across alternatives, or have an integer multiple structure.
The continuous time version of this procedure also satisfies the IZ guarantee, with a tight worst-case
preference-zone PCS bound.

A slight modification of the discrete-time BIZ procedure for unknown and/or heterogeneous sampling
variances is given below.  This algorithm takes a variable number of samples
from alternative $x$ at time $t$, and $n_{tx}$ is this number (its definition may found in the algorithm given below). 
This algorithm depends on a collection of integers $B_{1},\ldots,B_{k}$, $P^{*},c,\delta$
and $n_{0}$. Here, $n_{0}$ is the number of samples to use in the first
stage of samples, and $100$ is the recommended value for $n_{0}$ when the variances are unknown. The paramater
$B_{x}$ controls the number of samples taken from system $x$ in
each stage. To simplify our analysis, the procedure presented is a slight modification of the original BIZ 
procedure \cite{Frazier:BIZ} where $z \in \mbox{arg max}_{x\in A} \lambda^2_{x}$, instead 
of $z \in \mbox{arg min}_{x\in A} n_{tx} / \lambda^2_{x}$. According to numerical experiments on common cases, there is little difference in the PCS between the version of BIZ
we analyze and the version in \citeN{Frazier:BIZ}.

For each $t$, $x\in\left\{ 1,\ldots,k\right\} $, and subset $A\subset\left\{ 1,\ldots,k\right\} $,
we define a function
\[
q'_{tx}\left(A\right)=\mbox{exp}\left(\delta\beta_{t}\frac{Z_{tx}}{n_{tx}}\right)\left/\sum_{x'\in A}\mbox{exp}\left(\delta\beta_{t}\frac{Z_{tx'}}{n_{tx'}}\right),\right.\mbox{ }\beta_{t}=\frac{\sum_{x'\in A}n_{tx'}}{\sum_{x'\in A}\hat{\lambda}_{tx'}^{2}}
\]
where $\hat{\lambda}_{tx'}^{2}$ is the sample variance of all samples
from alternative $x$ thus far, and $Z_{tx}=Y_{n_{tx},x}$ is the sum of
the samples from alternative $x$ observed by stage $t$.

\paragraph{Algorithm: Discrete-time implementation of BIZ, for unknown and/or heterogeneous variances.}    
\begin{algorithmic}[1]   
\label{alg:hetero-BIZ}   
\REQUIRE $c \in [0,\cmax]$, $\delta>0$, $P^*\in(1/k,1)$, $n_0\ge0$ an integer, $B_1,\ldots,B_k$ strictly positive integers.  Recommended choices are $c=\cmax$, $B_1=\cdots=B_k=1$ and $n_0$ between $10$ and $30$.     If the sampling variances $\lambda^2_x$ are known, replace the estimators     
$\lambdahat^2_{tx}$ with the true values $\lambda^2_x$, and set $n_0=0$.     

\STATE For each $x$, sample alternative $x$ $n_0$ times and set $n_{0x} \leftarrow n_0$.     
Let $W_{0x}$ and $\lambdahat^2_{0x}$ be the sample mean and sample variance respectively of these samples.     Let $t\leftarrow 0$. Let $z \in \mbox{arg max}_{x\in A} \lambdahat^2_{x}$, where $\lambdahat^2_{x}$ is the empirical estimator of the variance $\lambda^2_x$ using $n_0$ samples if $x\in A$.    
\STATE Let $A \leftarrow \{ 1,\ldots, k\}$, $\upthresh \leftarrow P^*$.
\WHILE{$x\in\mbox{max}_{x\in A} q'_{tx}\left(A\right)<P$}
\WHILE{$\mbox{min}_{x\in A} q'_{tx}\left(A\right) \le c$}
 \STATE Let $x\in\mbox{arg min}_{x\in A} q_{tx}\left(A\right)$.
    \STATE Let $\upthresh \leftarrow \upthresh/(1-q_{tx}\left(A\right))$.     
\STATE Remove $x$ from $A$.
\ENDWHILE    
\STATE For each $x\in A$, let      $n_{t+1,x} = \ceil\left( \lambdahat^2_{tx} (n_{tz} + B_z) / \lambdahat^2_{tz} \right)$.     \STATE For each $x\in A$, if $n_{t+1,x}>n_{tx}$, take $n_{t+1,x}-n_{tx}$ additional samples from alternative $x$.  Let $W_{t+1,x}$ and $\lambdahat^2_{t+1,x}$ be the sample mean and sample variance respectively of all samples from alternative $x$ thus far.    
\STATE Increment $t$.
 \ENDWHILE
  \STATE Select $\xhat \in\mbox{arg max}_{x\in A} Z_{tx} / n_{tx}$ as our estimate of the best.

\end{algorithmic} 
\hspace{5 mm}

This algorithm generalizes the BIZ procedure with known common variance. In that case, we have that $B_1=\cdots=B_k=1$ and $n_{tx}=t$. The algorithm can be generalized to the continuous case (see  \citeN{Frazier:BIZ}). 

\section{ASYMPTOTIC VALIDITY WHEN THE VARIANCES ARE KNOWN}
\label{proof}

In this section we prove that the BIZ procedure satisfies asymptotically the IZ guarantee
when the variances are known. This means that we consider a collection of
ranking and selection problems parametrized by $\delta>0$.  
For the problem given $\delta$, we suppose that the vector of the true means
$\mu=\left[\mu_{1},\ldots,\mu_{k}\right]$ is equal to $\delta a$ for some fixed $a\in\mathbb{R}^{k}$ that
does not depend on $\delta$ and $a_{k}>a_{k-1}\geq \cdots \geq a_{1}$, $a_{k}-a_{k-1}>1$.
Moreover, the variances of the alternatives are finite, strictly greater than zero and do not depend on $\delta$.
We also suppose that samples from system $x\in\left\{ 1\ldots,k\right\} $ are identically distributed
and independent, over time and across alternatives. We also define
$\lambda_{z}^{2}:=\max_{i\in\left\{ 1\ldots,k\right\} }\lambda_{i}^{2}$.

Any ranking and selection algorithm can be viewed as mapping from
paths of the $k$-dimensional discrete-time random walk $\left(Y_{tx}:t\in\mathbb{N},x\in\left\{ 1,\ldots,k\right\} \right)$
onto selection decisions. 
Our proof uses this viewpoint, noting that 
the BIZ procedure's mapping from paths onto
selections decisions is the composition of three
simpler maps.

The first mapping is from the raw discrete-time random walk $\left(Y_{tx}:t\in\mathbb{N},x\in\left\{ 1,\ldots,k\right\} \right)$
onto a time changed version of this random walk, written as $\left(Z_{tx}:t\in\mathbb{N},x\in\left\{ 1,\ldots,k\right\} \right)$,
where we recall $Z_{tx}=Y_{n_{tx},x}$ is the sum of
the samples from alternative $x$ observed by stage $t$.

The second one maps this time-changed random walk through a non-linear
mapping for each $t,x$ and subset $A\subset\left\{ 1,\ldots,k\right\} $,
to obtain $\left(q_{tx}^{'}\left(A\right):t\in\mathbb{N},A\subset\left\{ 1,\ldots,k\right\} ,x\in A\right)$,
where

\[
q'_{tx}\left(A\right)=\mbox{exp}\left(\delta\beta_{t}\frac{Z_{tx}}{n_{tx}}\right)\left/\sum_{x'\in A}\mbox{exp}\left(\delta\beta_{t}\frac{Z_{tx}}{n_{tx}}\right)\right.:=q'\left(\left(Z_{tx}:x\in A\right),\delta,t\right)
\]
where we note that $n_{x}\left(t\right)$ and $\beta_{t}$ are deterministic
in the version of the known-variance BIZ procedure that we consider
here.

The third one maps the paths of $\left(q_{tx}^{'}\left(A\right):t\in\mathbb{N},A\subset\left\{ 1,\ldots,k\right\} ,x\in A\right)$
onto selection decisions. Specifically, this mapping begins with $A_0 = \{1,\ldots,k\}$,
$P_0=P^*$, 
and finds the first time $\tau_1$ that $q'_{tx}(A_0)$ falls above the threshold $P_0$,
or below the threshold $c$.
If the first case occurs, the alternative with the largest $q'_{\tau_1,x}(A_0)$ is selected as the best.
If the second case occurs, the alternative with the smallest
$q'_{\tau_1,x}(A_0)$ is eliminated, resulting in a new set $A_1$, a new selection threshold 
$P_1$ is calculated from $P_0$ and the eliminated alternative's value of $q'_{\tau_1,x}(A_0)$, and the process continues.
This process is repeated until an alternative is selected as the best.
Call this mapping $h$, so that the BIZ selection decision is
$h\left(\left(q_{tx}^{'}\left(A\right):t\in\mathbb{N},A\subset\left\{ 1,\ldots,k\right\} ,x\in A\right)\right)$.

\subsection{Proof Outline}
Based on this view of the BIZ procedure as a composition of three maps, we outline the main ideas of our proof here. 

Our proof first notes that the same selection decision is obtained if we apply
the BIZ selection map $h$ to a time-changed version of 
$\left(q_{tx}^{'}\left(A\right):t\in\mathbb{N},A\subset\left\{ 1,\ldots,k\right\} ,x\in A\right)$, specifically to 
\begin{equation*}
\left(q_{tx}\left(A\right):t\in\delta^2 \mathbb{N},A\subset\left\{ 1,\ldots,k\right\} ,x\in A\right),
\end{equation*}
where
$q_{tx}\left(A\right) := q'\left(\left(Z_{\frac{t}{\delta^{2}}x}:x\in A\right),\delta,t \right)$.

This discrete-time process is interpolated by the continuous-time process 
\begin{equation}
\left(q_{tx}\left(A\right):t\in\mathbb{R},A\subset\left\{ 1,\ldots,k\right\} ,x\in A\right).
\label{eq:q-continuous}
\end{equation}
If we apply the BIZ selection map $h$ to this continuous-time process, the selection decision will differ from BIZ's selection decision for $\delta>0$, but we show that this difference vanishes as $\delta\to0$.  Thus, our proof focuses on showing that, as $\delta\to0$, applying the BIZ selection map $h$ to \eqref{eq:q-continuous} produces a selection decision that satisfies the indifference-zone guarantee. 

To accomplish this, we use a functional central limit theorem for
$Z_{\frac{t}{\delta^{2}}x}$, which shows that a centralized
version of $Z_{\frac{t}{\delta^{2}}x}$ converges to a Brownian motion as $\delta$ goes to $0$.
This centralized version of $Z_{\frac{t}{\delta^{2}}x}$ is 
\[
\mathcal{C}_{x}\left(\delta,t\right):=\frac{Y_{n_{x}\left(t\right),x}-t\lambda_{x}^{2}\mu_{x}}{\frac{\lambda_{x}^{2}}{\lambda_{z}\delta}}.
\]
Rewriting $Z_{\frac{t}{\delta^{2}}x}$ in terms of $\mathcal{C}_{x}\left(\delta,t\right)$ and substituting into the definition of $q_{tx}(A)$ provides the expression
\begin{equation}
    q_{tx}\left(A\right)=q\left(\left(\mathcal{C}_{x}\left(\delta,t\right)\frac{\lambda_{x}^{2}}{\delta\lambda_{z}^{2}}+\frac{\lambda_{x}^{2}}{\lambda_{z}^{2}}\left(n_{0}+\frac{t}{\delta^{2}}\right)\delta a_{x}:x\in A\right),\delta,t \right).
    \label{eq:q}
\end{equation}

We will construct a mapping $f\left(\cdot,\delta\right)$ that takes as input the process
$\left(\mathcal{C}_{x}\left(\delta,t\right) : x\in \{1,\ldots,k\}, t\in\mathbb{R}\right)$,
calculates 
\eqref{eq:q-continuous}
from it,
applies the BIZ selection map $h$ to \eqref{eq:q-continuous}, 
and then returns $1$ if the correct selection was made,
and $0$ otherwise.
Thus, the correct selection event that results from applying
the BIZ selection map $h$ to \eqref{eq:q-continuous} is the result
of applying the mapping $f\left(\cdot,\delta\right)$
to the paths $t\mapsto\mathcal{C}_{x}\left(\delta,t\right)$ .

With these pieces in place, the last part of our proof is to observe that (1) $\mathcal{C}\left(\delta,\cdot\right)$
converges to a multivariate Brownian motion $W$ as $\delta$ goes
to 0; (2) the function $f$ has a continuity property that causes
\[
f\left(\mathcal{C}\left(\delta,\cdot\right),\delta\right)\Rightarrow g\left(W\right)
\]
where g is the selection decision from applying the BIZ procedure
in continuous time; and (3) the BIZ procedure satisfies the IZ guarantee
when applied in continuous time (Theorem~1 in \citeN{Frazier:BIZ}), and so $E[g(W)] \ge P^*$ with equality for the worst configurations in the preference zone.

\subsection{Preliminaries for the Proof of the Main Theorem}

In this section, we present preliminary results and definitions used in the proof of the main theorem: first, a central limit theorem Corollary~\ref{cor:CLT}; second, definitions of the functions $f(\cdot,\delta)$ and $g(\cdot)$; third, a continuity result Lemma~\ref{l:continuity}; and fourth, a result Lemma~\ref{l:interpolation} that allows us to change from discrete-time processes to continuous-time processes.

First, we are going to see that the centralized sum of the output
data $\mathcal{C}_x(\delta,t)$ converges to a Brownian motion in the sense of $D_{\infty}:=D[0,\infty)$,
which is the set of functions from $\left[0,\infty\right)$ to $\mathbb{R}$
that are right-continuous and have left-hand limits, with the Skorohod
topology. The definition and the properties of this topology may be
found in Chapter 3 of \citeN{billingsley:convergnce2}. 

We briefly recall the definition of convergence of random paths in
the sense of $D_{\infty}$. Suppose that we have a sequence of random
paths $\left(\mathcal{X}_{n}\right)_{n\geq0}^{\infty}$ such that
$\mathcal{X}_{n}:\varOmega\rightarrow D_{\infty}$ where $\left(\Omega,\mathcal{F},\mathbb{P}\right)$
is our probability space. We say that $\mathcal{X}_{n}\Rightarrow\mathcal{X}_{0}$
in the sense of $D_{\infty}$ if $P_{n}\Rightarrow P_{0}$ where $P_{n}:\mathcal{D}_{\infty}\rightarrow\left[0,1\right]$
are defined as $P_{n}\left[A\right]=\mathbb{P}\left[\mathcal{X}_{n}^{-1}\left(A\right)\right]$
for all $n\geq0$ and $\mathcal{D}_{\infty}$ are the Borel subsets
for the Skorohod topology.

The following lemma shows that the centralized sum of the output
data with $t$ changed by $t/\delta^{2}$ converges to a Brownian
motion in the sense of $D_{\infty}$. 

\begin{lemma}
Let $x\in\left\{ 1\ldots,k\right\} $, then
\[
\mathcal{C}_{x}\left(\delta,\cdot\right)\Rightarrow W_{x}\left(\cdot\right)
\]
as $\delta\rightarrow0$ in the sense of $D[0,\infty)$, where $W_{x}$
is a standard Brownian motion.
\end{lemma}

\begin{proof}
By Theorem 19.1 of  \citeN{billingsley:convergnce2},

\[
\frac{Y_{n_{x}\left(t\right),x}-\mbox{floor}\left(\frac{\lambda_{x}^{2}}{\lambda_{z}^{2}}\left(\cdot\frac{1}{\delta^{2}}\right)\right)\mu_{x}}{\frac{\lambda_{x}^{2}}{\lambda_{z}}\sqrt{\frac{1}{\delta^{2}}}}\Rightarrow W_{x}\left(\cdot\right)
\]
in the sense of $D[0,\infty)$. 

Fix $w\in\Omega$. Observe that
\[
\frac{Y_{\mbox{floor}\left(\frac{\lambda_{x}^{2}}{\lambda_{z}^{2}}\left(\cdot\frac{1}{\delta^{2}}\right)\right),x}-\mbox{floor}\left(\frac{\lambda_{x}^{2}}{\lambda_{z}^{2}}\left(\cdot\frac{1}{\delta^{2}}\right)\right)\mu_{x}}{\frac{\lambda_{x}^{2}}{\lambda_{z}}\sqrt{\frac{1}{\delta^{2}}}}-\frac{Y_{\mbox{ceil}\left(\frac{\lambda_{x}^{2}}{\lambda_{z}^{2}}\left(\cdot\frac{1}{\delta^{2}}\right)\right),x}-\mbox{ceil}\left(\frac{\lambda_{x}^{2}}{\lambda_{z}^{2}}\left(\cdot\frac{1}{\delta^{2}}\right)\right)\mu_{x}}{\frac{\lambda_{x}^{2}}{\lambda_{z}}\sqrt{\frac{1}{\delta^{2}}}}\rightarrow0
\]
uniformly in $\left[0,s\right]$ for all $s\geq0$ and then by Theorem
A.2 
\[
\frac{Y_{\mbox{ceil}\left(\frac{\lambda_{x}^{2}}{\lambda_{z}^{2}}\left(\cdot\frac{1}{\delta^{2}}\right)\right),x}-\mbox{ceil}\left(\frac{\lambda_{x}^{2}}{\lambda_{z}^{2}}\left(\cdot\frac{1}{\delta^{2}}\right)\right)\mu_{x}}{\frac{\lambda_{x}^{2}}{\lambda_{z}}\sqrt{\frac{1}{\delta^{2}}}}\Rightarrow W_{x}\left(\cdot\right)
\]
in the sense of $D[0,\infty)$. 

Since $\frac{\frac{\lambda_{x}^{2}}{\lambda_{z}^{2}}t\frac{1}{\delta^{2}}-ceil\left(\frac{\lambda_{x}^{2}}{\lambda_{z}^{2}}t\frac{1}{\delta^{2}}\right)}{\frac{\lambda_{x}^{2}}{\lambda_{z}}\sqrt{\frac{1}{\delta^{2}}}}\rightarrow0$
uniformly on $[0,s]$ for every $s\geq0$, then by Theorem A.2

\[
\frac{Y_{\mbox{ceil}\left(\frac{\lambda_{x}^{2}}{\lambda_{z}^{2}}\left(\cdot\frac{1}{\delta^{2}}\right)\right),x}-\mbox{}\left(\frac{\lambda_{x}^{2}}{\lambda_{z}^{2}}\left(\cdot\frac{1}{\delta^{2}}\right)\right)\mu_{x}}{\frac{\lambda_{x}^{2}}{\lambda_{z}}\sqrt{\frac{1}{\delta^{2}}}}\Rightarrow W_{x}\left(\cdot\right).
\]

Finally, observe that for fixed $\omega\in\Omega$,
\begin{eqnarray*}
 &  & \frac{Y_{\mbox{ceil}\left(\frac{\lambda_{x}^{2}}{\lambda_{z}^{2}}\left(\cdot\frac{1}{\delta^{2}}\right)\right),x}-\mbox{}\left(\frac{\lambda_{x}^{2}}{\lambda_{z}^{2}}\left(\cdot\frac{1}{\delta^{2}}\right)\right)\mu_{x}}{\frac{\lambda_{x}^{2}}{\lambda_{z}}\sqrt{\frac{1}{\delta^{2}}}}-\frac{Y_{\mbox{ceil}\left(\frac{\lambda_{x}^{2}}{\lambda_{z}^{2}}\left(\cdot\frac{1}{\delta^{2}}\right)+n_{0}\frac{\lambda_{x}^{2}}{\lambda_{z}^{2}}\right),x}-\mbox{}\left(n_{0}\frac{\lambda_{x}^{2}}{\lambda_{z}^{2}}+\frac{\lambda_{x}^{2}}{\lambda_{z}^{2}}\left(\cdot\frac{1}{\delta^{2}}\right)\right)\mu_{x}}{\frac{\lambda_{x}^{2}}{\lambda_{z}}\sqrt{\frac{1}{\delta^{2}}}}\\
 & = & \frac{Y_{\mbox{ceil}\left(\frac{\lambda_{x}^{2}}{\lambda_{z}^{2}}\left(\cdot\frac{1}{\delta^{2}}\right)\right),x}-Y_{\mbox{ceil}\left(\frac{\lambda_{x}^{2}}{\lambda_{z}^{2}}\left(\cdot\frac{1}{\delta^{2}}\right)+n_{0}\frac{\lambda_{x}^{2}}{\lambda_{z}^{2}}\right),x}+\left(n_{0}\frac{\lambda_{x}^{2}}{\lambda_{z}^{2}}\right)\mu_{x}}{\frac{\lambda_{x}^{2}}{\lambda_{z}}\sqrt{\frac{1}{\delta^{2}}}}\\
 & \rightarrow & 0
\end{eqnarray*}
uniformly in $\left[0,t\right]$ for all $t\geq0$, and so by Theorem
A.2 the result follows.

\end{proof}

Now, we use the product topology in $D^{k}\left[0,\infty\right)$
for $k\in\mathbb{N}$. This topology may be described as the one under
which $\left(Z_{n}^{1},\ldots,Z_{n}^{k}\right)\rightarrow\left(Z_{0}^{1},\ldots,Z_{0}^{k}\right)$
if and only if $Z_{n}^{i}\rightarrow Z_{0}^{i}$ for all $i\in\left\{ 1,\ldots,k\right\} $.
See the Miscellany of  \citeN{billingsley:convergence}. The following corollary follows
from the previous result and independence.

\begin{corollary}
    \label{cor:CLT}
We have that
\[
\mathcal{C}\left(\delta,\cdot\right):=\left(\mathcal{C}_{x}\left(\delta,\cdot\right)\right)_{x\in A}\Rightarrow W\left(\cdot\right):=\left(W_{x}\left(\cdot\right)\right)_{x\in A}
\]
as $\delta\rightarrow0$ in the sense of $D_{\infty}^{k}$.
\end{corollary}

Now that we have obtained this functional central limit theorem for 
$\mathcal{C}\left(\delta,\cdot\right)$, we now continue along the proof outline and define the function $f(\cdot,\delta)$ that was sketched there.
This function has three parts: first, computing a ``non-centralized'' path from an arbitrary input ``centralized'' path in $D\left[0,\infty\right)^{k}$; second, applying the BIZ selection map $h$ to this non-centralized path; and third, reporting whether selection was correct or not.

To accomplish the first part, for each $F\in D\left[0,\infty\right)^{k}$, 
we define $q_{tx}^{F,\delta}\left(A\right)$ as
\[
q_{tx}^{F,\delta}\left(A\right)=q'\left(\left(F_{x}\left(t\right)\frac{\lambda_{x}^{2}}{\delta\lambda_{z}^{2}}+\frac{\lambda_{x}^{2}}{\lambda_{z}^{2}}\left(n_{0}+\frac{t}{\delta^{2}}\right)\delta a_{x}:x\in A\right),\delta,A\subset\left\{ 1,\ldots,k\right\} \right).
\]
Note that if we replace $F$ by $\mathcal{C}\left(\delta,t\right)$,
we get $q_{tx}\left(A\right)$ in \eqref{eq:q}.

To accomplish the second and third parts, we define
$f\left(F,\delta\right)$ to be obtained by applying the BIZ selection map $h$ to the process
$\left(q_{tx}^{F,\delta}\left(A\right):t\in\mathbb{R},A\subset\left\{ 1,\ldots,k\right\} ,x\in A\right)$, and then reporting whether the selection was correct.
More precisely, $f(F,\delta)$ is defined to be
\begin{eqnarray*}
f\left(F,\delta\right) & = & \begin{cases}
    1 & \text{if $h\left(\left(q_{tx}^{F,\delta}\left(A\right):t\in\mathbb{R},A\subset\left\{ 1,\ldots,k\right\} ,x\in A\right)\right) = k$,}\\
0 & \text{otherwise.}
\end{cases}
\end{eqnarray*}

We now construct a function $g(\cdot)$ that, when applied to the path of a $k$-dimensional standard Brownian motion, will be equal in distribution to the indicator of the correct selection event from the continuous-time BIZ procedure from \citeN{Frazier:BIZ} 
to a transformed problem that does not depend on $\delta$.

We construct $g$ analogously to $f(\cdot,\delta)$, but we replace the path $q_{tx}^{F,\delta}$ used in the construction of $f(\cdot,\delta)$ by a new path $q_{tx}^{F}$ that doesn't depend on $\delta$, and is obtained by taking the limit as $\delta\to0$.  This path is
\[
q_{tx}^{F}\left(A\right):=\mbox{exp}\left(\frac{F_{x}\left(t\right)}{\lambda_{z}}+\frac{1}{\lambda_{z}^{2}}ta_{x}\right)/\sum_{x^{'}\in A}\mbox{exp}\left(\frac{F_{x'}\left(t\right)}{\lambda_{z}}+\frac{1}{\lambda_{z}^{2}}ta_{x^{'}}\right).
\]
Then, $g$ is defined to be
\begin{eqnarray*}
g\left(F\right) & = & \begin{cases}
    1 & \text{if $h\left(\left(q_{tx}^{F}\left(A\right):t\in\mathbb{R},A\subset\left\{ 1,\ldots,k\right\} ,x\in A\right)\right) = k$,}\\
0 & \text{otherwise.}
\end{cases}
\end{eqnarray*}

In the proof of the main theorem, we will show that 
\[
f\left(\mathcal{C}\left(\delta,\cdot\right),\delta\right)\Rightarrow g\left(W\right)
\]
as $\delta\rightarrow0$ in distribution. We will use the following lemma, which shows a continuity property.
A proof of Lemma~\ref{l:continuity} may be found in a full version of this paper \cite{fullpaper},
which will be submitted soon to arXiv. 

\begin{lemma}
Let $\left\{ \delta_{n}\right\} \subset\left(0,\infty\right)$ such
that $\delta_{n}\rightarrow0$. If $D_{s}\equiv\{Z\in D\left[0,\infty\right)^{k}:\mbox{ if }\left\{ Z_{n}\right\} \subset D\left[0,\infty\right)^{k}\mbox{ and }$
$\mbox{lim}{}_{n}d_{\infty}\left(Z_{n},Z\right)=0$ , then the sequence
$\left\{ f\left(Z_{n},\delta_{n}\right)\right\} $ converges to $\left\{ g\left(Z\right)\right\} $$\left.\right\} $$ $,
then $ $$\mathbb{P}\left(W\text{\ensuremath{\in}}D_{s}\right)=1$.
\label{l:continuity}
\end{lemma}

The following lemma shows that the difference in the correct selection events obtained from applying the BIZ selection map $h$ to the discrete-time and continuous-time versions of $q_{tx}(A)$ vanish as $\delta$ goes to $0$.
A proof of Lemma~\ref{l:interpolation} may be found in a full version of this paper \cite{fullpaper}.

\begin{lemma}
    $\lim_{\delta\rightarrow0} \mathbb{P}\left( h\left( \left(q_{tx}^{'}\left(A\right):t\in\mathbb{N},A\subset\left\{ 1,\ldots,k\right\} ,x\in A\right) \right) =k\right) = \lim_{\delta\rightarrow0}
 \mathbb{P}\left(
 f\left(\mathcal{C}\left(\delta,t\right),\delta\right)=1
 \right)$.

\label{l:interpolation}
\end{lemma}

\subsection{The Main Result}

\begin{theorem}
\label{t:main}
If samples from system $x\in\left\{ 1\ldots,k\right\} $ are identically
distributed and independent, over time and across alternatives, then
$\mbox{lim}_{\delta\rightarrow0}\mbox{PCS}(\delta) \geq P*$
provided $\mu_{k}=a_{k}\delta,\mu_{k-1}=a_{k-1}\delta,\ldots,\mu_{1}=a_{1}\delta$,
$a_{k}>a_{k-1}\geq \cdots \geq a_{1}$, $a_{k}-a_{k-1}\ge1$, and the variances are finite and
do not depend on $\delta$.

Furthermore,
\[
\inf_{a\in PZ\left(1\right)}\mbox{lim}_{\delta\rightarrow0}\mbox{PCS}(\delta)=P^{*}
\]
where $PZ\left(1\right)=\left\{ a\in\mathbb{R}^{k}:a_{k}-a_{k-1}>1, a_{k}>a_{k-1}\geq \cdots \geq a_{1} \right\} $. 

\end{theorem}

\begin{proof}
Using the definitions given at the beginning of this section,
the selection decision of the discrete-time BIZ procedure for a particular $\delta>0$ when $\mu_{k}=a_{k}\delta,\mu_{k-1}=a_{k-1}\delta,\ldots,\mu_{1}=a_{1}\delta$ is given by
\begin{equation*}
h\left(\left(q_{tx}^{'}\left(A\right):t\in\mathbb{N},A\subset\left\{ 1,\ldots,k\right\} ,x\in A\right)\right)
\end{equation*}
and the probability of correct selection $\mbox{PCS}(\delta)$ is
\begin{equation*}
    \mbox{PCS}(\delta) =  \mathbb{P}\left( h\left( \left(q_{tx}^{'}\left(A\right):t\in\mathbb{N},A\subset\left\{ 1,\ldots,k\right\} ,x\in A\right) \right) =k\right).
\end{equation*}
By Lemma~\ref{l:interpolation}, we have that 
\begin{equation}
\lim_{\delta\to0}
\mbox{PCS}(\delta)
 = \lim_{\delta\rightarrow0}
 \mathbb{P}\left(
 f\left(\mathcal{C}\left(\delta,t\right),\delta\right)=1
 \right).
 \label{eq:convergence1}
\end{equation}

We also have, by Lemma~\ref{l:continuity} and an extension of the continuous mapping theorem (Theorem~5.5 of \citeN{billingsley:convergence}), 
\begin{equation*}
    f\left(\mathcal{C}\left(\delta,t\right),\delta\right)\Rightarrow g\left(W\left(t\right)\right)
\end{equation*}
in distribution as $\delta\rightarrow0$.
This implies that
\begin{equation}
\lim_{\delta\rightarrow0}\mathbb{P}\left(f\left(\mathcal{C}\left(\delta,t\right),\delta\right)=1\right)
 = \mathbb{P}\left(g\left(W\right)=1\right).
 \label{eq:convergence2}
\end{equation}

The random variable $g(W)$ is equal in distribution to the indicator of the event of correct selection that results from applying the continuous-time BIZ procedure from \citeN{Frazier:BIZ} in a problem with indifference-zone parameter equal to 1, where each alternative's observation process has volatility $\lambda_z$ and drift $a_x$.
This can be seen by noting that the path $(q^{W}_{tx}(A) : t\ge0)$ defined above is equal in distribution to the path $(q_{tx}(A) : t\ge0)$ defined in equation (2) of \citeN{Frazier:BIZ}, and that the selection decision of the continuous-time algorithm in \citeN{Frazier:BIZ} is obtained by applying $h$ to this path.

Theorem 1 in \citeN{Frazier:BIZ} states that 
\begin{equation}
\mathbb{P}\left(g\left(W\right)=1\right)\ge P^{*}.
\label{eq:correct}
\end{equation}

Combining 
\eqref{eq:convergence1}, 
\eqref{eq:convergence2}, and
\eqref{eq:correct}, we have
\begin{equation*}
    \lim_{\delta\to0} \mbox{PCS}(\delta) \ge P^*.
\end{equation*}

Furthermore, Theorem~1 in \citeN{Frazier:BIZ} shows that 
\begin{equation}
\mbox{inf}_{a\in PZ\left(1\right)}\mathbb{P}\left(g\left(W\right)=1\right)=P^{*}
\label{eq:tight}
\end{equation}
where $PZ\left(1\right)=\left\{ a\in\mathbb{R}^{k}:a_{k}-a_{k-1}\geq1\right\} $.

Combining 
\eqref{eq:convergence1}, 
\eqref{eq:convergence2}, and
\eqref{eq:tight}, shows
\begin{equation*}
\inf_{a\in PZ\left(1\right)}\lim_{\delta\rightarrow0}\mbox{PCS}(\delta)=P^{*}.
\end{equation*}

\end{proof}

\section{NUMERICAL EXPERIMENTS}
\label{numericalExperiments}

We now use simulation experiments to illustrate and further investigate the phenomenon characterized by Theorem~\ref{t:main}. Using the version of BIZ described in Section~\ref{sec:BIZ} with maximum elimination ($c=1-(P^*)^\frac{1}{k-1}$), we estimate and then plot the PCS as a function of $\delta$.
In all examples, $P^*=0.9$, PCS was estimated using 10,000 independent replications, and confidence intervals have length at most $0.014$.

\begin{figure}[h]
  \centering
   \begin{subfigure}[b]{0.3\textwidth}
      \centering
       \includegraphics[width=\textwidth]{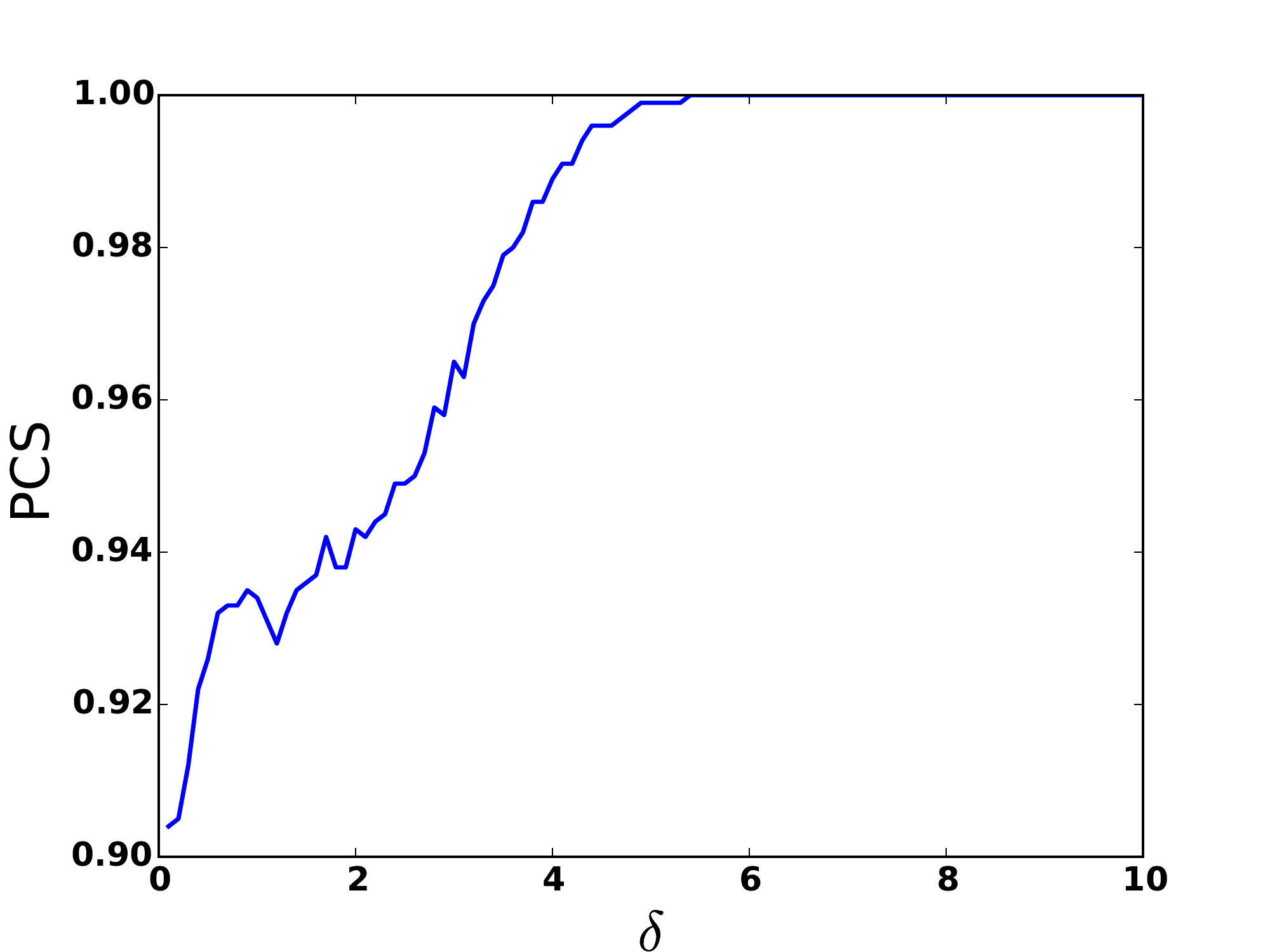}
    \caption{Known heterogeneous variances, 
        $\lambda^2_k = 0.25$, $\lambda_{1}^2 = 1$, $n_0=0$.
    \label{fig: tahi3}}
   \end{subfigure}
    \hfill
    \begin{subfigure}[b]{0.3\textwidth}
        \centering
         \includegraphics[width=\textwidth]{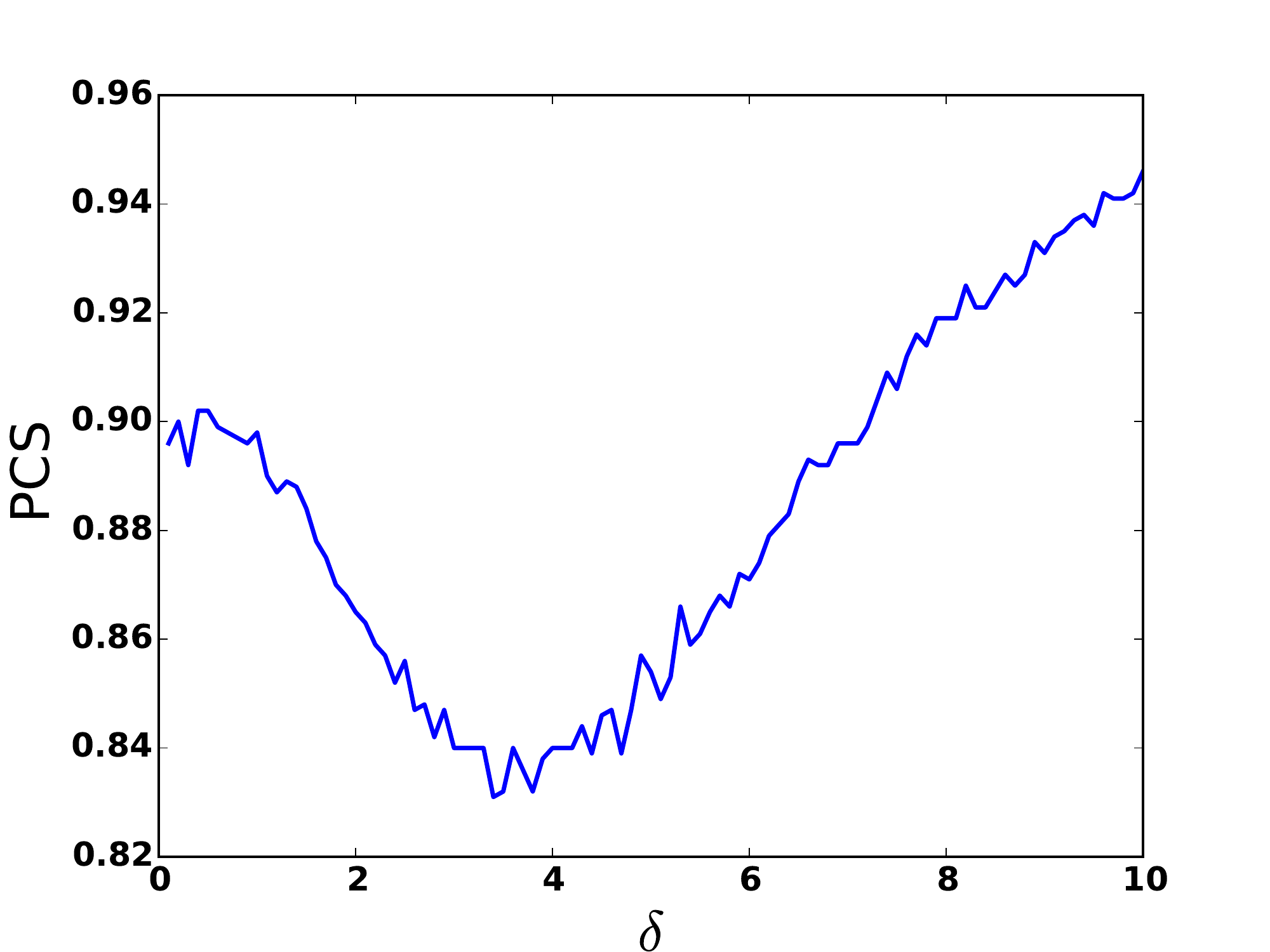}
             \caption{
    Unknown highly heterogeneous variances, 
    $\lambda^2_{k} = 100$, $\lambda_{1}^2 = 1$, $n_0=15$.
    \label{fig: tahi2}}
    \end{subfigure}
        \hfill
    \begin{subfigure}[b]{0.3\textwidth}
        \centering
         \includegraphics[width=\textwidth]{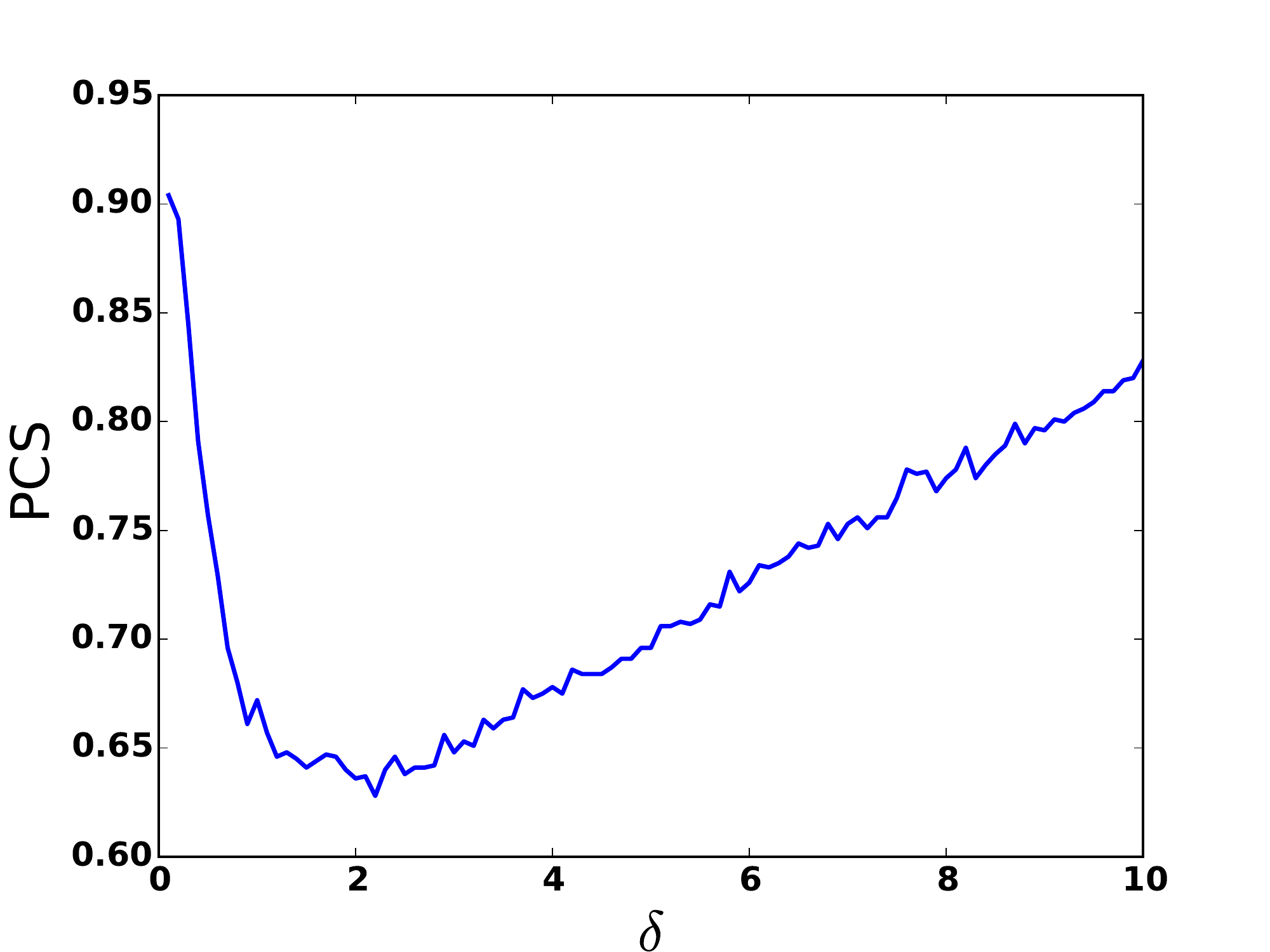}
    \caption{ 
    Known highly heterogeneous variances, 
    $\lambda^2_{k} = 100$, $\lambda_{1}^2 = 1$, $n_0=0$.
        \label{fig: tahi}}
    \end{subfigure}
\caption{
    The PCS of the BIZ procedure versus $\delta$ for three different slippage configurations with 100 alternatives and $P^*=0.9$.
    We observe in all three examples that the PCS converges to $P^*$ as $\delta$ goes to $0$.
    The first example (a) shows typical behavior, where the PCS is above $P^*$ for all values of $\delta$.
    The second (b) and third (c) examples are atypical, and were chosen specially to illustrate that BIZ can underdeliver on PCS in slippage configurations when $n_0$ is small and the variance of the best alternative is much larger than the variance of the other alternatives.
\label{fig:tahi10}}
\end{figure}  

Our first example, illustrated in Figure~\ref{fig: tahi3}, is a known variance slippage configuration where the variance of the best alternative is 1/4 of the variance of the worst alternative. Specifically, we consider $100$ systems with independent normally distributed samples, where $\mu_{k}=\delta,\mu_{k-1}=0,\dots,\mu_{1}=0$, $\delta$ is within the interval $[0.1,10]$, and $\lambda_{100}=1,\lambda_{99}=1+\frac{(0.5)(98)}{99},\cdots,\lambda_{1}=0.5$. Here, $n_{0}=0$. Figure~\ref{fig: tahi3} shows that in this example the IZ guarantee is always satisfied.  Moreover, the PCS approaches $P^*$ as $\delta$ goes to zero, as predicted by Theorem~\ref{t:main}. When $\delta$ is big enough, the PCS is almost one because the difference between the best system and the others is large enough to be easily identifiable by the BIZ procedure. 

Our second example, illustrated in Figure~\ref{fig: tahi2}, is an unknown variance slippage configuration where the variance of the best alternative is 100 times larger than the variance of the other alternatives. Although Theorem~\ref{t:main} applies only to the known-variance version of BIZ, we conjecture that the unknown-variance version of BIZ should exhibit similar behavior.
In this example, we consider $100$ systems with independent normally distributed samples, where $\mu_{100}=\delta,\mu_{99}=0,\dots,\mu_{1}=0$, $\delta$ is within the interval $[0.1,10]$, and $\lambda_{100}=10,\lambda_{99}=\cdots=\lambda_{1}=1$. We set $n_{0}=15$.
As $\delta$ goes to $0$, we observe that the PCS converges to $P^*$, as it did in the known-variance slippage configuration example.
In this example, we have intentionally chosen $n_0$ to be smaller than the recommended value of 100, and have chosen a large variance for the best system, to cause BIZ to fail to meet the IZ guarantee for $\delta>0$.
Increasing the parameter $n_0$ typically causes BIZ to meet the IZ guarantee for all $\delta$, and we recommend a larger value of $n_0$ in practice.  The choice of $n_0$, and its impact on PCS, merits further study.

Our third example, illustrated in Figure~\ref{fig: tahi}, uses the same sampling distributions as the second example, but assumes the variances are known, and sets $n_0=0$.
The effect of this change, and especially of setting $n_0$ to $0$, is to cause significant underdelivery on PCS for large values of $\delta$.  
As remarked above, this example was chosen specially to illustrate that BIZ can underdeliver on PCS in slippage configurations when $n_0$ is small, and the variance of the best alternative is much larger than the variance of the worst alternative.  However, as predicted by Theorem~\ref{t:main}, the PCS converges to $P^*$ as $\delta$ grows small, even in this pathological case.

\section{CONCLUSION}
\label{conclusion}

We have proved the asymptotic validity of the Bayes-inspired Indifference Zone
procedure \cite{Frazier:BIZ} when the variances are known.
This algorithm has been observed empirically to take fewer samples than other IZ procedures, especially for
problems with large numbers of alternatives, and so characterizing when it
satisfies the indifference-zone guarantee is important for understanding when
it should be used in practice.

\section*{ACKNOWLEDGMENTS}
Peter Frazier and Saul Toscano-Palmerin were partially supported by NSF CAREER CMMI-1254298, AFOSR FA9550-12-1-0200, and AFOSR FA9550-15-1-0038.
Peter Frazier was also partially supported by NSF IIS-1247696 and the Atkinson Center for a Sustainable Future Academic Venture Fund.
Saul Toscano-Palmerin was also partially supported by the Mexican Secretariat of Public Education (SEP).

\bibliographystyle{chicago}
\bibliography{Untitled}

\section*{AUTHOR BIOGRAPHIES}

\noindent {\bf SAUL TOSCANO-PALMERIN} is a Ph.D. student of the School of Operations Research and Information Engineering at Cornell University, Ithaca, NY. His research interest is in simulation optimization, machine learning and sequential decision-making under uncertainty. His email address is st684@cornell.edu. \\

\noindent {\bf PETER I. FRAZIER} is an Associate Professor of the School of Operations Research and Information Engineering at Cornell University, Ithaca, NY.  He holds a Ph.D. in operations research and financial engineering industrial engineering from Princeton University. His research interests include optimal learning, sequential decision-making under uncertainty, and machine learning, focusing on applications in simulation optimization, design of experiments, materials science, e-commerce and medicine. He is the secretary of the INFORMS Simulation Society.
His e-mail address is pf98@cornell.edu.\\

\end{document}